\documentclass[11pt]{article}

\usepackage[a4paper,margin=1in]{geometry}
\usepackage{amsmath,amssymb,amsthm,mathtools}
\usepackage{microtype}
\usepackage{booktabs}
\usepackage{array}
\usepackage{enumitem}
\usepackage{algorithm}
\usepackage{algpseudocode}
\usepackage{xcolor}
\usepackage{hyperref}
\usepackage{cleveref}

\hypersetup{
  pdftitle={A Simple Factor-5 Conversion to Internal Collage Systems},
  pdfauthor={Simone Faro},
  colorlinks=true,
  linkcolor=blue!55!black,
  citecolor=blue!55!black,
  urlcolor=blue!55!black
}

\newtheorem{theorem}{Theorem}
\newtheorem{lemma}[theorem]{Lemma}
\newtheorem{corollary}[theorem]{Corollary}
\newtheorem{remark}[theorem]{Remark}
\theoremstyle{definition}
\newtheorem{definition}[theorem]{Definition}
\newtheorem{example}[theorem]{Example}

\newcommand{\C}{\mathcal{C}}
\newcommand{\Cp}{\mathcal{C}'}
\newcommand{\val}[1]{\langle #1\rangle}

\title{A Factor-5 Conversion to Internal Collage Systems}
\author{Simone Faro\\
Department of Mathematics and Computer Science, University of Catania\\
Viale Andrea Doria 6, 95125 Catania, Italy\\
\texttt{faro@dmi.unict.it}\\
\href{https://orcid.org/0000-0001-5937-5796}{ORCID: 0000-0001-5937-5796}}
\date{}

\begin{document}
\maketitle

\begin{abstract}
A collage system is a grammar-based compression model that extends straight-line programs with repetition and substring truncation. Internal collage systems additionally require every nonterminal to be structurally reachable from the start symbol. Migita, Uehata, and I (CPM 2026) showed that any collage system of size $m$ can be converted into an internal one generating the same string with size at most $9m$, and left the improvement of this constant as an open problem. We show that a simple refinement of their top-down conversion reduces the bound to $5m$. The proof combines three elementary ideas: canonicalization of generated truncations that remain aligned with a target endpoint; a repetition decomposition that absorbs every complete copy of the repetition base into a maximal core; and monotonicity of structural reachability, which prevents an input truncation rule from both becoming structurally reachable and later acting as a hidden target at which endpoint alignment is lost. The conversion runs in deterministic $O(m^2)$ worst-case time in the stated unit-cost random-access machine model. Consequently, for every string, the minimum size of an internal collage system is at most five times the minimum size of a general collage system.
\end{abstract}

\section{Introduction}

Grammar-based compression describes a string by a compact acyclic collection of production rules rather than by listing its characters explicitly. Straight-line programs (SLPs) are the most classical example: every nonterminal produces either one character or the concatenation of two previously defined nonterminals, and the start symbol derives exactly one string. The minimum grammar size is therefore both a compression objective and a measure of structural repetitiveness~\cite{Navarro2021Survey}; for a survey on algorithms over SLP-compressed strings, see Lohrey~\cite{Lohrey2012}.

Collage systems, introduced by Kida et al.~\cite{KidaEtAl2003}, enrich this grammar model with repetition and truncation. In the original formulation, substring extraction is implemented through prefix and suffix truncations. Following the internal-collage model of Navarro, Ochoa, and Prezza~\cite{NavarroOchoaPrezza2021} and the general formulation adopted by Migita, Uehata, and I~\cite{MigitaEtAl2026}, throughout this paper we treat one arbitrary substring truncation as a single production. Repetition and truncation can represent repetitive strings more flexibly than ordinary SLPs and have made the minimum collage-system size $c(T)$ a natural repetitiveness measure in the study of dictionary compression~\cite{Navarro2021Survey,NavarroOchoaPrezza2021}.

Truncation, however, creates a structural phenomenon that does not arise in ordinary SLPs. A reachable nonterminal may use only an internal substring of another nonterminal $X$. The value of $X$ is then needed to define the text, even though $X$ itself need never appear as an exposed component in the expansion of the start symbol. Such a target is hidden behind a truncation edge.

Internal collage systems, introduced by Navarro, Ochoa, and Prezza~\cite{NavarroOchoaPrezza2021}, exclude this behavior. Informally, a collage system is internal if every nonterminal is structurally reachable from the start symbol without first entering it through a truncation edge. Let $\widehat c(T)$ denote the minimum size of an internal collage system generating $T$. Since internal systems form a subclass of general collage systems, $c(T)\le \widehat c(T)$ is immediate. The nontrivial question is how large the cost of enforcing internality can be.

Migita, Uehata, and I recently settled the asymptotic question~\cite{MigitaEtAl2026}. They showed that any collage system of size $m$ can be transformed in $O(m^2)$ time into an internal collage system generating the same string and having size at most $9m$. Their proof processes hidden original nonterminals in top-down order. A truncation referring to the current hidden target is either redirected to a smaller target or decomposed into a small number of pieces. The difficult point is that a Case-3 decomposition may create new truncation rules that themselves refer to hidden targets and therefore generate future work. Bounding this recursively generated work leads to the factor $9$. Migita et al. explicitly list improving the ``probably loose'' constant $9$ among the main open directions left by their result~\cite{MigitaEtAl2026}.

The main result of this note is that every collage system of size $m$ can be converted into an internal collage system generating the same string with size at most $5m$. Equivalently,
$\widehat c(T)\le5c(T)$.
The conversion runs in deterministic $O(m^2)$ worst-case time in the stated unit-cost RAM model.

The proof keeps the same top-down strategy but makes generated truncations more disciplined. As long as a generated truncation continues to represent a prefix or suffix of its current target, we regard it as \emph{clean}. Clean prefixes of the same target are nested, and so are clean suffixes. Hence it is enough to keep one longest live prefix representative and one longest live suffix representative; all shorter representatives, as well as ties at maximum length, can be redirected internally through the chosen one. This simple canonicalization ensures that only a constant number of clean lineages can lose endpoint alignment at any one target.

A second refinement concerns repetitions. When a requested substring crosses several copies of the repetition base, every complete copy is absorbed into a maximal central core, leaving only proper boundary pieces. For a clean generated lineage, one boundary is absent, which reduces the local helper cost.

The final ingredient is a distinction between two roles played by original truncation rules. Some become structurally reachable and act as active referrers. Others remain hidden and may later act as \emph{breakers}, namely original truncation targets through which a clean generated lineage loses endpoint alignment. Structural reachability is monotone during the conversion, so the same original truncation cannot play both roles. This disjointness provides the extra saving needed to reach factor $5$.

The resulting proof is intentionally elementary. It uses local helper bounds, canonical endpoint representatives, monotone reachability, and guaranteed deletions of hidden original targets. No potential function, frontier DAG, or global conservation argument is required. The complete counting argument, including the execution-dependent counters, is deferred to Section~\ref{sec:counting}.

As in Migita et al.~\cite{MigitaEtAl2026}, the running time is quadratic in the grammar size. We state the bound in a unit-cost RAM model in which all repetition exponents, substring endpoints, rule indices, and represented lengths manipulated by the conversion fit in $O(1)$ machine words. We do not identify this word size with $\Theta(\log |T|)$, since intermediate represented lengths may be much larger than the length of the final text.

Throughout the paper we use the model in which an arbitrary substring truncation $X\to Y[b..e)$ is one production, as in Navarro, Ochoa, and Prezza~\cite{NavarroOchoaPrezza2021} and Migita, Uehata, and I~\cite{MigitaEtAl2026}. In the original formulation of Kida et al.~\cite{KidaEtAl2003}, substring extraction is instead implemented through prefix and suffix truncations; Migita et al. explicitly note this distinction~\cite{MigitaEtAl2026}. These formulations are expressively equivalent for generating strings, but that fact alone does not imply rule-for-rule equivalence, preservation of internality, or preservation of the same constant-factor bounds.

The rest of the paper is organized as follows. Section~\ref{sec:prelim} fixes the model and internality convention. Section~\ref{sec:previous} recalls the factor-$9$ conversion. Section~\ref{sec:conversion} presents the refinements used in the factor-$5$ argument. Section~\ref{sec:counting} gives the global counting proof, and Section~\ref{sec:time} establishes deterministic $O(m^2)$ running time.

\clearpage
\section{Preliminaries}
\label{sec:prelim}

Let $\Sigma$ be a finite alphabet and let $T=T[1]\cdots T[n]$ be a string over $\Sigma$. We use one-based half-open substring notation $T[b..e)=T[b]T[b+1]\cdots T[e-1]$, with $1\le b<e\le n+1$.

A \emph{collage system} is an acyclic collection of nonterminals, each with exactly one production of one of the following forms:
\[
\begin{array}{rcll}
X&\to&a, & a\in\Sigma,\\
X&\to&YZ, & \text{concatenation},\\
X&\to&Y^r,\quad r>2, & \text{repetition},\\
X&\to&Y[b..e), & \text{substring truncation}.
\end{array}
\]
The value derived by $X$ is denoted by $\val{X}$. The system has one distinguished start symbol whose value is the represented text. Its size is the number of nonterminals. We assume that the dependency relation is acyclic and fix a topological order in which every right-hand-side nonterminal precedes the nonterminal that refers to it.

The restriction $r>2$ is conventional: a double repetition is represented by a concatenation $YY$. Following the model used by Migita et al.~\cite{MigitaEtAl2026}, one arbitrary substring truncation counts as one rule. This point matters for constant-factor statements.

\begin{definition}[Structural reachability]
Starting from the start symbol, atomic, concatenation, and repetition rules expose the nonterminals they reference. A truncation $Q\to X[b..e)$ uses the substring $\val{X}[b..e)$ but does not by itself expose $X$ as a structurally reachable nonterminal. A nonterminal is \emph{reachable} if it is exposed from the start symbol through this structural expansion; otherwise it is \emph{unreachable}.
\end{definition}

A collage system is \emph{internal} if every nonterminal is reachable. This is equivalent to the grammar-tree definition used in~\cite{NavarroOchoaPrezza2021,MigitaEtAl2026}. We write $c(T)=\min\{|\C|:\C\text{ generates }T\}$ and $\widehat c(T)=\min\{|\C|:\C\text{ is internal and generates }T\}$.

\begin{example}
Consider $A\to a$, $B\to b$, $X\to AB$, and $S\to X[1..2)$. The represented text is $a$. The start symbol is reachable, but the truncation edge does not expose $X$, so $X$, $A$, and $B$ are initially unreachable. Internalization rewrites the start rule to $S\to a$ and deletes the now-unused original nonterminals.
\end{example}

The conversion repeatedly rewrites truncation rules whose left-hand side belongs to the structurally reachable part of the current system while their original target is still hidden.

\begin{definition}[Live truncation]
A current truncation rule
\[
   Q\to X[b..e)
\]
is \emph{live} if and only if $Q$ is structurally reachable, $X$ is an original nonterminal, and $X$ is structurally unreachable.
\end{definition}

It is useful to follow a live rule through repeated Case-1 redirects.

\begin{definition}[Lineage]
An \emph{original truncation lineage} starts from a truncation rule present in the input when that rule becomes live. A \emph{generated lineage} starts from one of the at most two live boundary truncations created by a Case-3 decomposition. Case-1 redirects preserve the identity of the lineage. A lineage terminates when it is resolved by a terminal case or, more generally, as soon as its current truncation ceases to be live. In particular, termination occurs when the target becomes reachable and the truncation becomes internal, or when canonicalization redirects the lineage to a structurally reachable representative.
\end{definition}

\section{The previous factor-$9$ conversion}
\label{sec:previous}

We briefly recall the conversion of Migita, Uehata, and I~\cite{MigitaEtAl2026}. Original unreachable nonterminals are processed in decreasing topological order. When an unreachable target $X$ is processed, every surviving reference to $X$ from the reachable part of the system is a truncation rule
\[
   Q\to X[b..e).
\]
The rewrite depends on the original production of $X$ and on the position of the requested interval.

\begin{description}[leftmargin=2.2em,style=nextline]
\item[Case 0.] $X$ is atomic. The truncation is replaced by the corresponding atomic rule.
\item[Case 1.] The requested interval lies inside one smaller component, or $X$ itself is defined by truncation. The truncation is redirected to a smaller original target. No helper nonterminal is created.
\item[Case 2.] $X\to Y^r$ and the requested interval spans more than two copies of $Y$. The interval is written as a suffix of $Y$, followed by a block of whole copies of $Y$, followed by a prefix of $Y$. At most four new nonterminals are needed.
\item[Case 3.] The requested interval crosses exactly one component boundary of a concatenation or two consecutive copies of a repetition. All crossing referrers of the same target are processed together. A longest required suffix and a longest required prefix are shared, and the size increases by at most two per participating referrer.
\end{description}

The crucial observation in Case 3 is that only the shared longest suffix and longest prefix can remain live and point to unreachable targets~\cite{MigitaEtAl2026}. Therefore each original concatenation or repetition target creates at most two new live truncation lineages.

If the input has size $m$ and contains $a$ original truncation rules, the analysis of~\cite{MigitaEtAl2026} counts at most $2(m-a)$ generated truncation rules and charges the generic Case-2 cost $4$ to every original or generated truncation lineage. This gives $m+4(a+2(m-a))=9m-4a\le9m$. The present note keeps the same four-case decomposition. The only change is that generated truncations are treated more carefully before they are charged the generic Case-2 cost.

\section{A refined top-down conversion}
\label{sec:conversion}

\subsection{Clean and tainted generated lineages}

A generated live truncation is born in Case 3 as either a prefix or a suffix of its current target. As long as this endpoint structure is preserved, the lineage is cheaper to process than an arbitrary truncation.

\begin{definition}[Clean and tainted]
A generated live lineage is \emph{clean} if, continuously since its creation, it has represented a prefix or a suffix of every current target along the lineage. It becomes \emph{tainted} at the first Case-1 redirect that destroys this endpoint alignment. Taint is permanent.
\end{definition}

A Case-1 redirect through a concatenation or repetition target preserves endpoint alignment whenever the lineage remains inside one smaller component. Only composition through an original truncation rule can move a clean endpoint into the interior of a new target.

\begin{lemma}[Taint origin]
A generated clean lineage can become tainted only when Case 1 redirects it through an original truncation target.
\end{lemma}

\begin{proof}
Let the current clean interval be a prefix or suffix of $X$. If $X\to YZ$, a clean prefix that lies in one component must lie in $Y$ and remains a prefix of $Y$; the suffix case is symmetric. If $X\to Y^r$, a clean prefix contained in one copy must lie in the first copy and remains a prefix of $Y$, while a clean suffix lies in the last copy. By contrast, if $X\to Y[b..e)$, a prefix of $X$ can be mapped to an interval $Y[b..c)$ with $b>1$, and similarly for a suffix. Thus only a redirect through an original truncation target can destroy alignment.
\end{proof}

All clean prefixes pointing to the same target are nested, and so are all clean suffixes. This gives a simple canonicalization.

\begin{lemma}[Canonical clean representatives]
\label{lem:canonical}
Fix an unreachable original target $X$. Among all live generated clean prefixes of $X$, choose one representative of maximum length and keep only that one live; every other clean prefix, including another prefix of the same maximum length, can be rewritten as an internal truncation of the chosen representative. The same holds for suffixes. If one or more full-length clean representatives exist, choose one of them and keep only that representative live. These rewrites preserve acyclicity.
\end{lemma}

\begin{proof}
Let $Q^\star$ be a chosen maximum-length clean prefix of $X$, and let $Q\neq Q^\star$ be any other clean prefix. Since clean prefixes of the same target are nested, $|\val{Q}|\le|\val{Q^\star}|$. Before canonicalization, both current productions are truncations of the same original target $X$: in particular, the dependency graph contains the edges $Q^\star\to X$ and $Q\to X$. Rewrite
\[
   Q\to Q^\star[1..|\val{Q}|+1).
\]
If $|\val{Q}|=|\val{Q^\star}|$, this is simply a full-length truncation of $Q^\star$. Because $Q^\star$ is structurally reachable, the new truncation is internal and the lineage of $Q$ terminates.

It remains to check acyclicity. The current collage system is acyclic before the rewrite. If adding the new dependency $Q\to Q^\star$ created a cycle, then before the rewrite there would already have been a directed path $Q^\star\leadsto Q$. Since the current production of $Q^\star$ is a truncation of $X$, every directed path starting at $Q^\star$ begins with the edge $Q^\star\to X$; hence such a path would imply $X\leadsto Q$. But the old edge $Q\to X$ is still present before the rewrite, so $Q\to X\leadsto Q$ would already be a directed cycle, contradicting acyclicity. Therefore no path $Q^\star\leadsto Q$ exists and the canonicalization rewrite preserves acyclicity.

The suffix case is identical. If clean full-length representatives exist, every such representative is simultaneously a full-length prefix and suffix of $X$. Choose one, say $F$. Every other clean prefix or suffix $Q$ can be rewritten as the corresponding truncation of $F$; in particular, another full-length representative is rewritten as a full-length truncation of $F$. The same acyclicity argument applies because, before the rewrite, both $F$ and $Q$ have current truncation edges to $X$. Thus after canonicalization there is literally at most one live clean prefix and at most one live clean suffix, or a single live full-length representative. No new nonterminal is introduced.
\end{proof}

The top-down order has a second consequence that will be used both for canonicalization and for the bucket invariant below. Fix the original topological order so that every original nonterminal occurring on the right-hand side of an original production is strictly smaller than its left-hand side.

\begin{lemma}[Monotone target descent and no late arrivals]
\label{lem:no-late-arrivals}
Whenever the current target of a live lineage changes, its new original target is strictly smaller in the fixed original topological order. Moreover, every new live lineage created while an original target $X$ is being processed has an original target strictly smaller than $X$. Consequently, once the processing of $X$ begins, no operation performed at $X$ or at any later target can create or redirect a live lineage whose current target is again $X$.
\end{lemma}

\begin{proof}
Consider first an already existing live lineage. In Case 1, the requested interval is composed through the original production of its current target. If the target is a concatenation or repetition, the redirected interval lies in one of the original children; if the target is itself a truncation, the composition redirects the interval to the original target of that truncation rule. In every case the new target occurs on the right-hand side of an original production and is therefore strictly smaller than the old target.

Now consider lineages created while an original target $X$ is processed. Case 0 creates no live lineage. Case 2 exposes the repetition base before its boundary truncations are used, so those boundary truncations are internal rather than live. Case 3 can create at most the shared boundary lineages, and each of them truncates an original child of $X$; hence its target is strictly smaller than $X$.

The remaining source of a new original live lineage is structural activation. Suppose that an activation performed while processing $X$ makes an original truncation left-hand side $Q$ structurally reachable. The activation starts from an original child of $X$ and propagates only through original concatenation and repetition edges, so $Q$ is a strict descendant of $X$ in the original dependency DAG and is therefore strictly smaller than $X$. If the original rule of $Q$ is $Q\to Z[b..e)$, then $Z$ is smaller than $Q$, and hence also strictly smaller than $X$. Thus the newly live lineage of $Q$ cannot target $X$ or any larger original nonterminal.

After $X$ has been completed, every later processed target is itself strictly smaller than $X$. Repeating the same argument shows that all subsequent redirects and all newly created live targets are smaller still. Therefore no live lineage can ever arrive at $X$ after its processing has begun.
\end{proof}

The conversion canonicalizes clean referrers before any Case-1 redirect through the current target. Hence at most two clean lineages can lose alignment at one original truncation target.

\begin{lemma}[At most two taints per truncation target]
\label{lem:two-taints}
During the processing of one original truncation target, at most two generated clean lineages become tainted.
\end{lemma}

\begin{proof}
Immediately before the redirect through the target, Lemma~\ref{lem:canonical} leaves literally at most one live clean prefix and one live clean suffix, or a single full-length representative; ties at maximum length have already been collapsed. By Lemma~\ref{lem:no-late-arrivals}, this canonical set is final: no live lineage can arrive at the target after its processing has begun. Therefore at most two clean lineages can lose alignment at that target.
\end{proof}

\begin{example}[Canonicalization, tainting, and a breaker]
\label{ex:breaker}
Consider an intermediate checkpoint at which an original unreachable target $X$ is defined by
\[
   X\to Z[3..10),
   \qquad |\val{Z}|=12.
\]
Thus $|\val{X}|=7$. Suppose that earlier Case-3 events have produced three structurally reachable generated clean lineages currently pointing to $X$:
\[
   A\to X[1..3),\qquad
   B\to X[1..5),\qquad
   C\to X[5..8).
\]
The first two are clean prefixes of $X$, while $C$ is a clean suffix. Canonicalization keeps the longer prefix $B$ and rewrites
\[
   A\to B[1..3),
\]
which is internal because $B$ is reachable. Hence only $B$ and $C$ remain live.

Processing the original truncation rule of $X$ composes the two surviving intervals with $X\to Z[3..10)$:
\[
   B\to Z[3..7),
   \qquad
   C\to Z[7..10).
\]
Both intervals lie strictly inside $\val{Z}$, so both lineages lose endpoint alignment and become permanently tainted. The target $X$ is therefore a breaker and is deleted after its bucket is exhausted.

Since $X$ is unreachable at the breaker event, monotonicity of structural reachability implies that $X$ was never reachable earlier. Consequently, its own original truncation rule never contributed a live original truncation lineage. This is exactly the role separation used later in Lemma~\ref{lem:lineage-breaker-disjoint}.
\end{example}

\subsection{A maximal-core Case-2 decomposition}

Starting from the suffix--core--prefix Case-2 decomposition of Migita, Uehata, and I~\cite{MigitaEtAl2026}, we impose a maximal-core convention in which every complete copy of the repetition base is absorbed into the core. This is the one place where clean alignment changes the local helper bound.

Let $X\to Y^r$, let $L=|\val{Y}|$, and consider an interval $X[b..e)$ spanning more than two copies of $Y$. Define
\[
s=
\begin{cases}
0,&(b-1)\bmod L=0,\\
L-((b-1)\bmod L),&\text{otherwise},
\end{cases}
\qquad
p=(e-1)\bmod L,
\]
and set $k=(e-b-s-p)/L$. Then
\[
   X[b..e)=S\,Y^k\,P,
\]
where $S$ is either empty or a proper suffix of $Y$, and $P$ is either empty or a proper prefix of $Y$. The core is maximal: a boundary equal to a whole copy of $Y$ is absorbed into $Y^k$.

If at least one boundary piece is nonempty, the core is represented directly by $Y$ when $k=1$, by one concatenation helper $R\to YY$ when $k=2$, and by one repetition helper $R\to Y^k$ when $k>2$. Proper nonempty boundary pieces need one truncation helper each. The existing referrer is then used as the final concatenation node, with at most one additional concatenation helper when all three components are present.

There is one degenerate but important subcase. If both boundary pieces are empty, then the requested substring is exactly $Y^k$. Because Case 2 applies only to intervals spanning more than two copies of $Y$, necessarily $k\ge3$. In this situation no core helper is created: the existing referrer itself is rewritten directly as
\[
   Q\to Y^k.
\]
This avoids an inadmissible unary production $Q\to R$ and strictly reduces the helper cost.

\begin{lemma}[Case-2 helper bounds]
\label{lem:case2}
A generic Case-2 referrer needs at most four new helper nonterminals. If the referrer belongs to a generated clean lineage, at most two new helpers are needed.
\end{lemma}

\begin{proof}
For an arbitrary interval with at least one nonempty boundary, there are at most two proper boundary helpers, at most one core helper, and at most one additional concatenation helper, for a total of four. If both boundaries are empty, the referrer is rewritten directly as $Q\to Y^k$ with $k\ge3$, so no helper is needed at all. Thus the generic four-helper bound still holds.

If the generated clean lineage is a prefix of $X$, then $b=1$ and hence $s=0$. If it is a suffix, then $e=|\val{X}|+1$ and hence $p=0$. Therefore at least one boundary is genuinely absent. If both are absent, the direct rewrite $Q\to Y^k$ uses no helper. Otherwise the decomposition has exactly two nonempty components: the maximal core and one proper boundary piece. The existing referrer combines them, so at most one core helper and one boundary helper are new.
\end{proof}

\begin{example}[Why the maximal core matters]
Let $|\val{Y}|=2$, let $X\to Y^5$, and consider the clean prefix $Q\to X[1..8)$. Its value has length $7$. Under the maximal-core convention, $X[1..8)=Y^3Y[1..2)$. Thus one core helper and one proper-prefix helper suffice, with the existing $Q$ used as the final concatenation node. Treating the first full copy of $Y$ as a separate boundary component would introduce an unnecessary extra concatenation helper and would destroy the two-helper clean bound.
\end{example}

\begin{example}[Core-only Case 2]
Let $Y\to\mathtt a$, $X\to Y^5$, and let the reachable referrer be $Q\to X[1..4)$. Here $s=p=0$ and $k=3$, so the requested value is exactly $Y^3$. The correct rewrite is directly
\[
   Q\to Y^3.
\]
Introducing a helper $R\to Y^3$ would leave only the unary assignment $Q\to R$, which is not a production type of the collage-system model.
\end{example}

\subsection{Case 3 and live generated lineages}

For Case 3 we use the shared construction of Migita et al.~\cite{MigitaEtAl2026}. All crossing referrers of the same target are processed together. The longest required suffix and longest required prefix are materialized once, and every participating referrer is rewritten through them.

We fix one additional convention that is relevant to reachability. If one of these shared boundary pieces coincides with an entire original child $Y$, we do \emph{not} replace it by a direct reference to $Y$. Instead we still materialize it as the full-length truncation helper
\[
   S\to Y[1..|\val{Y}|+1)
\]
(or symmetrically $P\to Y[1..|\val{Y}|+1)$). Thus Case 3 never makes an original child structurally reachable merely because a requested boundary happens to cover that child completely. We call this the \emph{full-length shielding convention}.

\begin{lemma}[Case-3 bound with full-length shielding]
\label{lem:case3}
Suppose that $q$ live referrers participate in one Case-3 processing. Under the full-length shielding convention, at most $2q$ new nonterminals are created, at most two of them are live generated truncations after the event, and no original child is structurally activated merely because a shared boundary equals that whole child.
\end{lemma}

\begin{proof}
The standard shared Case-3 construction of~\cite{MigitaEtAl2026} materializes the longest required suffix and the longest required prefix once, and expresses shorter requests through those shared objects. This uses at most two new objects per participating referrer.

If a longest boundary is the whole original child $Y$, full-length shielding materializes it as
\[
   S\to Y[1..|\val{Y}|+1)
\]
(or symmetrically for the prefix side) rather than using a direct structural reference to $Y$. This occupies exactly the boundary-helper slot already counted for that side, so the $2q$ bound is unchanged. Because the helper is still a truncation rule, it does not structurally expose $Y$.

Every new object other than the shared longest suffix and prefix refers only to reachable generated objects and is internal. Hence at most those two shared boundary truncations can remain live after the event.
\end{proof}

If a repetition target has any Case-2 referrer, we process Case 2 first. The first such rewrite exposes the repetition base $Y$ structurally. Any remaining interval crossing exactly two copies can then be represented by at most two truncations of the now-reachable $Y$ and creates no new live lineage. This priority rule can only reduce the number of generated live lineages and will be assumed below.

\subsection{Operational form and correctness}

The algorithm maintains, for each unreachable original target, a bucket containing the live lineages currently pointing to that target. Original targets are processed in decreasing topological order.

\begin{algorithm}[H]
\caption{Factor-$5$ internalization}\label{alg:conversion}
\small
\begin{algorithmic}[1]
\Require Collage system $\C$ with original nonterminals in decreasing topological order
\State initialize structural reachability from the start symbol
\State register every reachable original truncation lineage in the bucket of its unreachable target
\For{each original target $X$ in decreasing topological order}
   \If{$X$ is reachable}
      \State \textbf{continue}
   \EndIf
   \If{the live bucket of $X$ is empty}
      \State delete $X$ and \textbf{continue}
   \EndIf
   \State canonicalize all generated clean referrers of $X$ and remove every lineage thereby made internal
   \State apply Cases 0 and 1 to all referrers to which they apply, removing every terminated lineage
   \State for every surviving Case-1 redirect, preserve \textsc{tainted} permanently; if a \textsc{clean} lineage loses endpoint alignment, mark it \textsc{tainted}
   \State move every redirected live lineage to the bucket of its new target
   \If{crossing referrers remain}
      \If{$X$ is a repetition target with at least one Case-2 referrer}
         \State process all Case-2 referrers using the maximal-core decomposition (directly rewriting a core-only referrer as $Q\to Y^k$), remove them from the live bucket, and immediately \textsc{Activate} the repetition base as one logical update
         \State resolve every remaining two-copy referrer internally and remove it from the live bucket
      \Else
         \State process all remaining crossings jointly by Case 3 using full-length shielding, and remove the input lineages from the live bucket
         \State register the at most two newly created live boundary lineages
      \EndIf
   \EndIf
   \State delete $X$ after its bucket has been exhausted
\EndFor
\State \Return the resulting system
\Statex
\Procedure{Activate}{$Y$}
   \State mark $Y$ reachable and propagate through original concatenation/repetition edges
   \State clear the live bucket of every original target that becomes reachable
   \State register every newly reachable original truncation lhs in the bucket of its unreachable target
\EndProcedure
\end{algorithmic}
\end{algorithm}

The next invariant separates the two roles that an original truncation nonterminal may play in the counting argument: it may become structurally reachable and act as a referrer, or it may remain hidden and later be processed as an unreachable target, but it cannot do both.

\begin{lemma}[Monotonicity of structural reachability]
\label{lem:reachability-monotone}
Once an original nonterminal becomes structurally reachable during Algorithm~\ref{alg:conversion}, it remains structurally reachable for the rest of the conversion.
\end{lemma}

\begin{proof}
We identify the structural witness that makes a nonterminal reachable and show that this witness is never removed.

First consider reachability present at initialization. It is obtained from the start symbol by following original concatenation and repetition edges. Every nonterminal on such a witness path is reachable. The production at the tail of each witness edge is therefore an original concatenation or repetition production of a reachable nonterminal. Algorithm~\ref{alg:conversion} rewrites only truncation referrers; it never rewrites such a structural production. Moreover, reachable original nonterminals are never deleted. Hence every initial witness path persists.

Now consider a nonterminal that becomes reachable later. By the full-length shielding convention, Case 3 does not structurally activate an original child. The only new structural activation is therefore the one performed by \textsc{Activate} after the first Case-2 rewrite on a repetition target. That rewrite terminates the live lineage of its reachable referrer and replaces its truncation production by a stable expansion containing the repetition base. In the core-only subcase this is the direct repetition $Q\to Y^k$; otherwise the base is exposed directly in the final concatenation or through the newly created core helper. The referrer is not rewritten again after this terminal event, and any core helper has a concatenation or repetition production that is likewise never rewritten. Thus the structural path witnessing the activation of the repetition base persists.

Procedure \textsc{Activate} then propagates reachability only through original concatenation and repetition productions of newly reached original nonterminals. Those nonterminals are strictly below the currently processed target and have not yet been processed as unreachable targets. Their structural productions are therefore still their original ones. Once they become reachable, Algorithm~\ref{alg:conversion} will never delete them, and, as above, their concatenation or repetition productions are never rewritten. Hence every propagation edge added to the reachability witness also persists.

No other operation makes an original nonterminal newly structurally reachable. Therefore every original nonterminal that ever becomes structurally reachable retains a permanent structural path from the start symbol.
\end{proof}

The live buckets are not merely an implementation device: at the stable checkpoints of the algorithm they exactly represent the unresolved dependencies from the structurally reachable part of the current grammar. A \emph{checkpoint} is a state immediately after initialization or after one complete logical update has finished, including all associated reachability and bucket updates. In particular, a Case-2 rewrite together with the immediately following call to \textsc{Activate} is one logical update; no bucket invariant is asserted in the transient state between those two actions.

\begin{lemma}[Live-bucket completeness and safe deletion]
\label{lem:bucket-complete}
At every checkpoint of Algorithm~\ref{alg:conversion}, for every unreachable original target $Y$, its live bucket contains exactly the live truncation lineages whose current target is $Y$. In particular, every surviving truncation from a structurally reachable left-hand side to an unreachable original target is represented by a lineage in the corresponding live bucket.

Immediately before an original target $X$ is processed, every original nonterminal strictly larger than $X$ has already either become structurally reachable or been deleted. Consequently, if $X$ is unreachable and its live bucket is empty, deleting $X$ preserves the string generated by the start symbol.
\end{lemma}

\begin{proof}
We prove the checkpoint invariant by induction over the complete logical updates of the algorithm. At initialization there are no generated lineages. Structural reachability is propagated from the start symbol through original concatenation and repetition edges, and every reachable original truncation left-hand side whose target is still unreachable is registered in that target's bucket. Thus the invariant holds at the initial checkpoint.

Canonicalization may make a noncanonical generated clean lineage internal by redirecting it to a reachable canonical representative; after that rewrite the terminated lineage is removed from the live bucket, so the next checkpoint is consistent. In Cases 0 and 1, all terminal lineages are first removed and every surviving redirect is then inserted into the bucket of its new unreachable target; a redirect whose new target is already reachable becomes internal and is inserted into no live bucket. Thus the invariant holds after the Case-0/Case-1 update block.

For Case 2, the rewrite of the participating referrers and the immediately following call to \textsc{Activate} are treated as a single logical update. The Case-2 lineages terminate and are removed. Activation may make additional original nonterminals structurally reachable; every bucket whose target becomes reachable is then cleared because all its entries cease to be live, while every newly reachable original truncation left-hand side is registered in the bucket of its still-unreachable original target. Hence the invariant is restored at the checkpoint after activation. Residual two-copy resolutions only terminate and remove their input lineages. Case 3 terminates its input lineages and, before the next checkpoint, explicitly registers each newly created boundary truncation that is live. Therefore every checkpoint contains exactly the live lineages indexed by their current unreachable original targets.

A structurally reachable rule cannot contain an ordinary concatenation or repetition reference to an unreachable original nonterminal, because such a reference would expose that target and make it structurally reachable. Therefore any surviving reference from the reachable part of the grammar to an unreachable original target must be a truncation edge and, by the checkpoint invariant, is represented in the corresponding live bucket.

Finally, every original nonterminal that could refer to the current target $X$ is strictly larger than $X$ in the original topological order and has therefore already been encountered in the decreasing scan. If such a predecessor survived and is structurally reachable, then a surviving truncation reference to $X$ would place a lineage in the bucket of $X$, while an ordinary structural reference would make $X$ reachable. If the predecessor remained unreachable, it was deleted when its own turn was processed. By Lemma~\ref{lem:no-late-arrivals}, no smaller target processed later can create a new live reference back to $X$. Hence an unreachable $X$ with empty live bucket has no surviving incoming dependency from the part of the grammar that determines the start string, so deleting it preserves that string.
\end{proof}

\begin{lemma}[Correctness]
Algorithm~\ref{alg:conversion} preserves the string generated by the start symbol and returns an internal collage system.
\end{lemma}

\begin{proof}
Cases 0--3 replace each referrer by an equivalent representation of the same substring; in the core-only Case-2 subcase the referrer itself is rewritten as the repetition $Q\to Y^k$. Canonicalization rewrites a shorter clean representative as a truncation of a longer reachable representative. Hence every rewrite preserves the start value.

If an unreachable target has an empty bucket, Lemma~\ref{lem:bucket-complete} shows that it has no surviving incoming dependency from the represented reachable structure and may be deleted. If its bucket is nonempty, the algorithm first terminates or redirects every live reference and then deletes the target. Lemma~\ref{lem:no-late-arrivals} guarantees that no later operation can recreate a live reference to a processed target.

At termination, every original nonterminal is therefore either structurally reachable or deleted. Every generated nonterminal is introduced in the expansion of a reachable referrer. A generated truncation lineage either remains live until it is processed by a terminal case or terminates earlier when its truncation becomes internal, including by target activation or canonicalization. Hence every retained nonterminal is structurally reachable and the resulting system is internal.
\end{proof}

\section{The factor-$5$ counting argument}
\label{sec:counting}

The proof reduces to a small number of counters. Let the input system have size $m$. Let $a$ be the number of original truncation rules, $\rho$ the number of original truncation lineages that are ever live, $g$ the number of generated live lineages created during the conversion, and $\tau$ the number of generated tainted lineages that terminate in Case 2. Let $s$ be the number of original Case-3 targets that create at least one live generated lineage, and let $b$ be the number of original truncation targets at which at least one generated clean lineage becomes tainted. We call the latter targets \emph{breakers}.

The monotonicity of structural reachability separates the two possible roles of an original truncation rule.

\begin{lemma}[Original-lineage/breaker disjointness]
\label{lem:lineage-breaker-disjoint}
With the notation above,
\[
   \rho+b\le a.
\]
\end{lemma}

\begin{proof}
Associate each original truncation rule with its left-hand side. An original truncation lineage can become live only when that left-hand side is structurally reachable while the current target of the truncation is unreachable. Thus every rule counted by $\rho$ has a left-hand side that is structurally reachable at some time during the conversion.

By contrast, if an original truncation nonterminal $X$ is counted as a breaker, then $X$ is processed as an \emph{unreachable} original target: a live clean generated lineage points to $X$, and its Case-1 composition through the original truncation rule of $X$ destroys endpoint alignment. Therefore $X$ is structurally unreachable at the breaker event. By Lemma~\ref{lem:reachability-monotone}, $X$ cannot have been structurally reachable earlier. Consequently, the original truncation rule with left-hand side $X$ never gives rise to a live original lineage.

No original truncation rule is therefore counted both among the $\rho$ live original lineages and among the $b$ breakers. Since there are $a$ original truncation rules in total, $\rho+b\le a$.
\end{proof}

We next bound the total number of helper nonterminals.

\begin{lemma}[Helper creation]
\label{lem:helpers}
If $H$ is the total number of new helper nonterminals created by Algorithm~\ref{alg:conversion}, then
\[
   H\le4\rho+2g+2\tau
      \le4(a-b)+2g+2\tau.
\]
\end{lemma}

\begin{proof}
Follow every truncation lineage until it terminates or becomes internal, and charge each newly created helper at the unique local event in which that helper is introduced.

For a Case-3 event with $q$ input lineages and $h_E$ new helpers, Lemma~\ref{lem:case3} gives $h_E\le2q$. Order the inputs arbitrarily and assign at most two of the $h_E$ helpers to each in turn. The total capacity is $2q$, so every Case-3 helper is assigned exactly once and no participating lineage receives more than two units of charge. This bookkeeping applies equally to original and generated inputs.

By definition, exactly $\rho$ original truncation lineages are ever live. Cases 0 and 1 create no helpers. If an original lineage terminates in Case 2, at most four helpers are created for its local decomposition; a residual two-copy resolution or a Case-3 event costs it at most two. Thus reserving four units for each of the $\rho$ live original lineages dominates all possibilities. Original truncation rules whose left-hand side never becomes structurally reachable generate no live original lineage and receive no charge. Hence original lineages contribute at most $4\rho$.

For each of the $g$ generated lineages reserve two baseline units. If such a lineage becomes internal before a terminal rewrite, these units are unused. A Case-3 event or residual two-copy resolution costs it at most two helpers, and by Lemma~\ref{lem:case2} a clean generated lineage terminating in Case 2 also needs at most two.

Only a tainted generated lineage terminating in Case 2 can require more than the two-unit baseline. There are $\tau$ such lineages, each requiring at most four helpers, so two additional units per such lineage suffice. Therefore
\[
   H\le4\rho+2g+2\tau.
\]
Lemma~\ref{lem:lineage-breaker-disjoint} gives $\rho\le a-b$, yielding the second inequality.
\end{proof}

The remaining bounds are direct consequences of the local Case-3 and taint limits.

\begin{lemma}[Source and taint bounds]
\label{lem:source-taint}
The counters satisfy
\[
   g\le2s,\qquad
   \tau\le2b,\qquad
   s\le m-a.
\]
In particular, $g\le2(m-a)$.
\end{lemma}

\begin{proof}
Only Case 3 creates new live generated lineages, and by Lemma~\ref{lem:case3} each original Case-3 target creates at most two. Since $s$ counts exactly the original Case-3 targets that create at least one such lineage, $g\le2s$.

Every lineage counted by $\tau$ must first become tainted at a breaker. By Lemma~\ref{lem:two-taints}, one breaker can taint at most two clean generated lineages; hence $\tau\le2b$.

Finally, every target counted by $s$ has an original concatenation or repetition production. None is an original truncation target, and therefore there are at most $m-a$ possible source targets. Thus $s\le m-a$.
\end{proof}

\begin{lemma}[Deletion credit]
\label{lem:deletion}
If $D$ is the number of original nonterminals deleted during the conversion, then
\[
   D\ge s+b.
\]
\end{lemma}

\begin{proof}
Every target counted by $s$ is unreachable when its Case-3 bucket is processed and is deleted after that bucket is resolved. Every breaker is likewise an unreachable original target when processed and is deleted after its bucket is resolved. The two sets are disjoint: source targets counted by $s$ have original concatenation or repetition productions, whereas breakers have original truncation productions. Hence at least $s+b$ distinct original nonterminals are deleted.
\end{proof}

\begin{theorem}[Factor-$5$ internalization]
\label{thm:main}
Let $\C$ be a collage system of size $m$ generating a nonempty string. Algorithm~\ref{alg:conversion} produces an internal collage system $\Cp$ generating the same string and satisfying
\[
   |\Cp|\le5m.
\]
More precisely, if $a$ is the number of original truncation rules and $b$ is the number of breaker targets, then
\[
   |\Cp|\le4m+a-b.
\]
\end{theorem}

\begin{proof}
Let $H$ be the number of helper nonterminals created and let $D$ be the number of original nonterminals deleted. Since the conversion starts from $m$ original nonterminals,
\[
   |\Cp|=m+H-D.
\]
By Lemmas~\ref{lem:helpers} and~\ref{lem:deletion},
\[
\begin{aligned}
|\Cp|
&\le m+4(a-b)+2g+2\tau-s-b.
\end{aligned}
\]
Using Lemma~\ref{lem:source-taint}, $g\le2s$ and $\tau\le2b$, and therefore
\[
\begin{aligned}
|\Cp|
&\le m+4(a-b)+4s+4b-s-b\\
&=m+4a+3s-b.
\end{aligned}
\]
Again by Lemma~\ref{lem:source-taint}, $s\le m-a$. Hence
\[
\begin{aligned}
|\Cp|
&\le m+4a+3(m-a)-b\\
&=4m+a-b\\
&\le5m,
\end{aligned}
\]
because $a\le m$.
\end{proof}

\begin{corollary}
For every nonempty string $T$, $\widehat c(T)\le5c(T)$.
\end{corollary}

\begin{remark}
The inequality $\rho+b\le a$ is the only point where monotonicity of structural reachability enters the counting itself: breaker targets do not consume the four-helper budget reserved for live original truncation lineages.
\end{remark}

\section{Running time}
\label{sec:time}

We use a unit-cost RAM model in which every numerical quantity manipulated by the conversion---rule indices, repetition exponents, substring endpoints, and represented lengths---fits in $O(1)$ machine words, and arithmetic and comparisons on those words take constant time. The word size is assumed large enough for the largest value encountered during the conversion; it is not assumed to be merely $\Theta(\log |T|)$.

\begin{theorem}
Algorithm~\ref{alg:conversion} runs in deterministic $O(m^2)$ worst-case time in the unit-cost RAM model above.
\end{theorem}

\begin{proof}
There are $\rho\le a$ original truncation lineages that are ever live and, by Lemma~\ref{lem:source-taint}, at most $g\le2(m-a)$ generated live lineages. Thus only $O(m)$ live lineages ever exist over the whole execution.

Maintain one bucket per original target. A lineage is stored only in the bucket of its current target. A Case-1 redirect moves it to a strictly smaller original target in the fixed topological order. Therefore one lineage undergoes at most $O(m)$ redirects, and the total number of lineage--target incidences is $O(m^2)$.

Canonicalization can be maintained incrementally by storing, for each bucket, the longest clean prefix and longest clean suffix currently known. Each clean arrival requires only constant-time endpoint and length comparisons. Alternatively, scanning the bucket when it is opened still charges work only to lineage--target incidences.

Structural reachability is maintained incrementally with a Boolean mark and a queue or stack. Each original structural edge is propagated only when its tail first becomes reachable. Empty-bucket tests and deletions are constant-time bookkeeping operations.

Cases 0--3 create only $O(1)$ local helpers per participating referrer, and Case 3 finds its longest required suffix and prefix by scanning the current bucket once. Therefore all work is dominated by the $O(m^2)$ lineage--target incidences, and the total deterministic worst-case running time is $O(m^2)$.
\end{proof}

\begin{remark}[Bit complexity]
If binary arithmetic costs are charged explicitly, let $\ell$ be the maximum bit length of any exponent, truncation endpoint, or represented length encountered during the conversion. The algorithm performs $O(m^2)$ unit-cost RAM operations, but its bit complexity includes the cost of arithmetic and comparison on $\ell$-bit integers. No $O(m^2)$ bit-complexity claim is made.
\end{remark}

\section{Discussion and conclusion}

Theorem~\ref{thm:main} improves the previously explicit conversion guarantee from factor $9$ to factor $5$ while retaining deterministic quadratic running time. The proof also gives the refined estimate $|\Cp|\le4m+a-b$, where the breaker count $b$ records an execution-dependent saving that is invisible in a purely rule-type-based count.

The analysis still contains deliberate slack. Every live original truncation lineage is allowed the full four-helper Case-2 budget, and every generated lineage receives a two-helper baseline even if it later becomes internal without creating helpers. It is therefore natural to ask whether the universal conversion constant can be reduced further. The present result shows that a substantial improvement already follows from local structural information and from distinguishing hidden breaker targets from active original referrers.

\section*{Acknowledgements}
The author gratefully acknowledges the Gruppo Nazionale per il Calcolo Scientifico (GNCS) of the Istituto Nazionale di Alta Matematica (INdAM) for its support.

\section*{Declaration on the Use of AI Tools}
Generative AI tools, including OpenAI's ChatGPT, were used during the preparation of this manuscript to assist with language refinement, organization of the exposition, bibliographic searches, and exploratory checks of intermediate arguments. All mathematical statements, proofs, citations, and final conclusions were critically reviewed by the author, who assumes full responsibility for the content of the manuscript.

\end{document}